\documentclass[journal]{IEEEtran}
\usepackage{ amssymb }
\usepackage{amsthm}
\newtheorem{theorem}{Theorem}

\usepackage{cite}

\ifCLASSINFOpdf
   \usepackage[pdftex]{graphicx}
   \graphicspath{{../pdf/}{../jpeg/}}
  \DeclareGraphicsExtensions{.pdf,.jpeg,.png}
\else
   \usepackage[dvips]{graphicx}
   \graphicspath{{../eps/}}
   \DeclareGraphicsExtensions{.eps}
\fi
\usepackage{xcolor}

\usepackage[cmex10]{amsmath}
\begin{document}
%
\title{\LARGE Measurement-based Close-in Path Loss Modeling with Diffraction for\\Rural Long-distance Communications}
%
%
%

\author{Jaedon~Park, Hong-Bae Jeon,~\IEEEmembership{Graduate Student Member,~IEEE,} \\
	Jungho Cho,~\IEEEmembership{Member,~IEEE,} and Chan-Byoung~Chae,~\IEEEmembership{Fellow,~IEEE}
	\thanks{This work was	supported by the Agency for Defense Development.}
	\thanks{J. Park and J. Cho are with the Agency for Defense Development, Daejeon, Korea (e-mail:
		\{jaedon2, jh.cho\}@add.re.kr).}
	\thanks{H.-B. Jeon and C.-B.~Chae (corresponding author) are with the School of Integrated Technology, Yonsei University, Seoul, 
		Korea (e-mail: \{hongbae08, cbchae\}@yonsei.ac.kr).}
		\thanks{(\textit{J. Park and H.-B. Jeon are co-first authors.})}
		}

%
%

\markboth{IEEE Wireless Communications Letters}
{Park \MakeLowercase{\textit{et al.}}: CI models with Diffraction}
%



\maketitle

\begin{abstract}

In this letter, we investigate rural large-scale path loss models based on the measurements in a central area of  South Korea (rural area) in spring. In particular, we develop new close-in (CI) path loss models incorporating a diffraction component. The transmitter used in the measurement system is located on a hill and utilizes omnidirectional antennas operating at 1400 and 2250~MHz frequencies. The receiver is also equipped with omnidirectional antennas and measures at positions totaling 3,858 (1,262 positions for LOS and 2,596 positions for NLOS) and 4,957 (1,427 positions for LOS and 3,530 positions for NLOS) for 1400 and 2250~MHz, respectively. This research demonstrates that the newly developed CI path loss models incorporating a diffraction component significantly reduce standard deviations (STD) and are independent of frequency, especially for LOS beyond the first meter of propagation, making them suitable for use with frequencies up to a millimeter-wave.

\end{abstract}


%
\IEEEpeerreviewmaketitle

%
%
%
%

\section{Introduction}
Path loss models that are applicable to frequencies below 6~GHz and various heights have been released by the 3rd Generation Partnership Project (3GPP) and International Telecommunication Union (ITU)~\cite{3gppTR36873_2017, ITURM2135_2009guidelines}. Furthermore, based on the previous models developed for frequencies below 6 GHz, 3GPP has developed new path loss models for frequencies above 6~GHz~\cite{3gppTR38900_2018, MacCartney2017Rural} and 0.5~GHz~\cite{3gppTR38901_2019}. These path loss models from 3GPP and ITU suggest modeling for line-of-sight/non-line-of-sight (LOS/NLOS) scenarios, applying to urban microcell, urban and rural macrocell, as well as indoor office environments, with LOS models applying for the breakpoint distance. Close-in (CI) path loss models have been widely used in diverse environments and have demonstrated good accuracy over a wide range of frequencies~\cite{rappaport2002wireless}. In contrast to other models such as the alpha-beta-gamma model, which include floating-offset parameters that are far removed from the physical principles of outdoor channel propagation, the CI models are based on the fundamental physical principles of signal propagation presented by Friis and Bullington~\cite{MacCartney2014Omni}. The path loss exponent (PLE) included in the CI model offers physical insights into path loss based on the propagation environment.


Consequently, numerous researchers have been drawn to the development of close-in (CI) path loss models. Among them, MacCartney and Samimi \textit{et al.} have produced omnidirectional CI path loss models from measurements at 28 and 73 GHz in Downtown Manhattan, where the path loss exponents (PLEs) for line-of-sight (LOS) at those frequencies were 2.1 and 2.0, respectively, and for non-line-of-sight (NLOS) at both frequencies were 3.4~\cite{MacCartney2014Omni,samimi2015Probabilistic}. Additionally, MacCartney \textit{et al.} have conducted field measurements at 73 GHz in rural Virginia, deriving CI models (LOS: PLE 2.16, NLOS: PLE 2.75) and height-dependent path loss exponent (CIH) models (LOS: PLE 2.31, NLOS: PLE 3.07), which indicate that the PLE may not be frequency-dependent~\cite{MacCartney2017Rural}. Sun \textit{et al.} have also provided CI path loss models, along with alpha-beta-gamma (ABG) and frequency-weighted path loss exponent (CIF) models, for frequencies ranging from 2 to 73 GHz and for urban macrocell (UMa), urban microcell (UMi), and indoor hotspot (InH) scenarios. The CI models for the UMa scenario had a LOS PLE of 2.0 and an NLOS PLE of 2.9, both of which were frequency-independent~\cite{Sun2016Investi}.

For a more accurate path loss analysis, ITU introduced diffraction loss models such as single knife-edge (sKE), and delta-Bullington (DB) \cite{ITURP526_2013, ITURP1812_2019}. The sKE diffraction loss model treats obstacles as a knife-edge with negligible thickness \cite{ITURP526_2013}, while the DB diffraction loss model combines the Bullington part of diffraction with spherical-Earth diffraction \cite{ITURP1812_2019}. 
{Based on the review of literature, it is reasonable to assume that the inclusion of diffraction effect in the CI path loss model, which was neglected in previous studies~\cite{MacCartney2017Rural, MacCartney2014Omni, Rappaport2015Wideband}, can improve the accuracy of the path loss model. Additionally, it is plausible to assume that in outdoor settings, the path loss exponent of the CI models might be frequency-independent beyond the initial meter of propagation loss.}


{This letter presents modified CI path loss models for the rural macrocell (RMa) scenario~\cite{3gpptech, 3gpp24, metis}, incorporating a diffraction component to the general form. To the best of our knowledge, our research team is the first to take into account the diffraction component in the CI path loss model. These models are based on measurements conducted at 1400 and 2250~MHz in northern Seoul.}
As the ITU is investigating new frequency bands with strong radio properties for International Mobile Telecommunications (IMT), including those below 10~GHz~\cite{ITU23} and in the 1$\sim$2~GHz range~\cite{FAITU}, our tested frequencies of 1400 and 2500~MHz can serve as valuable examples for further research into path loss modeling. Our proposed models offer improved accuracy compared to models without a diffraction component. Additionally, we observe that our proposed CI path loss model with diffraction exhibits frequency-independence beyond the first meter of propagation, unlike existing models in 3GPP~\cite{MacCartney2017Rural,3gpptech, 3gpp24}. Hence, we can conclude that our proposed model is suitable for use in upper-band spectrum ranges~\cite{HBFSO, hjtvt}, including millimeter-wave~\cite{map, IAB}, where diffraction loss becomes a significant issue~\cite{SA}.


\section{CI Path Loss Model with Diffraction}
The general CI path loss model is given by~\cite{MacCartney2017Rural, CIPL}
\begin{equation} \label{eq:eq_CIpathlossOrig}
\text{PL}^{\text{CI}}(f,d)[\mathrm{dB}] = 20 \log_{10}\Big(\frac{4\pi f d_0}{c}\Big) +  n_{\text{CI}} \cdot 10 \cdot \log_{10}\left(\frac{d}{d_0}\right) + \chi_{\sigma},
\end{equation}
where $d$ is the path distance in meters between the transmitter and receiver, $d_0$ is the close-in free space reference distance in meters, $f$ is the carrier frequency in Hz, $c$ is the speed of light in m/s, $n_{\text{CI}}$ is the PLE, and the Gaussian random variable $\chi_{\sigma}$ represents the shadow fading with zero-mean and STD $\sigma$ in dB \cite{Sun2016Investi, Rappaport2015Wideband, MacCartney2014Omni, samimi2015Probabilistic, Sulyman2016Directional, rappaport2015millimeter}. In ($\ref{eq:eq_CIpathlossOrig}$), the first term on the right-hand side is the free-space path loss at the reference distance $d_0$, which is given by
\begin{equation} \label{eq:eq_FSPLd0}
\text{FSPL}(f,d_0)[\mathrm{dB}] = 20 \log_{10}\Big(\frac{4\pi f d_0}{c}\Big).
\end{equation}
Therefore, the path loss model can be rewritten as 
\begin{equation} \label{eq:eq_CIpathloss}
\text{PL}^{\text{CI}}(f_c,d)[\mathrm{dB}] = 32.4 + 20 \log_{10}f_c + n_{\text{CI}} \cdot 10 \log_{10}d  + \chi_{\sigma}
\end{equation}
where $f_c$ is the carrier frequency in GHz. 

Adding the diffraction component to the general CI model, the CI path loss model with diffraction can be expressed as 
\begin{multline} \label{eq:eq_CIpathlossDiff}
\text{PL}^{\text{CI}}(f_c,d,\phi)[\mathrm{dB}] = 32.4 + 20 \log_{10}f_c + n_{\text{CI}} \cdot 10 \log_{10}d \\
+ \alpha \cdot \phi + \chi_{\sigma},
\end{multline}
where $\phi$ is the diffraction component in dB, and the model parameter for diffraction $\alpha$ is the coefficient showing the dependence of path loss on the diffraction.

\begin{figure}[t]
\centering
\includegraphics[width=1.6in]{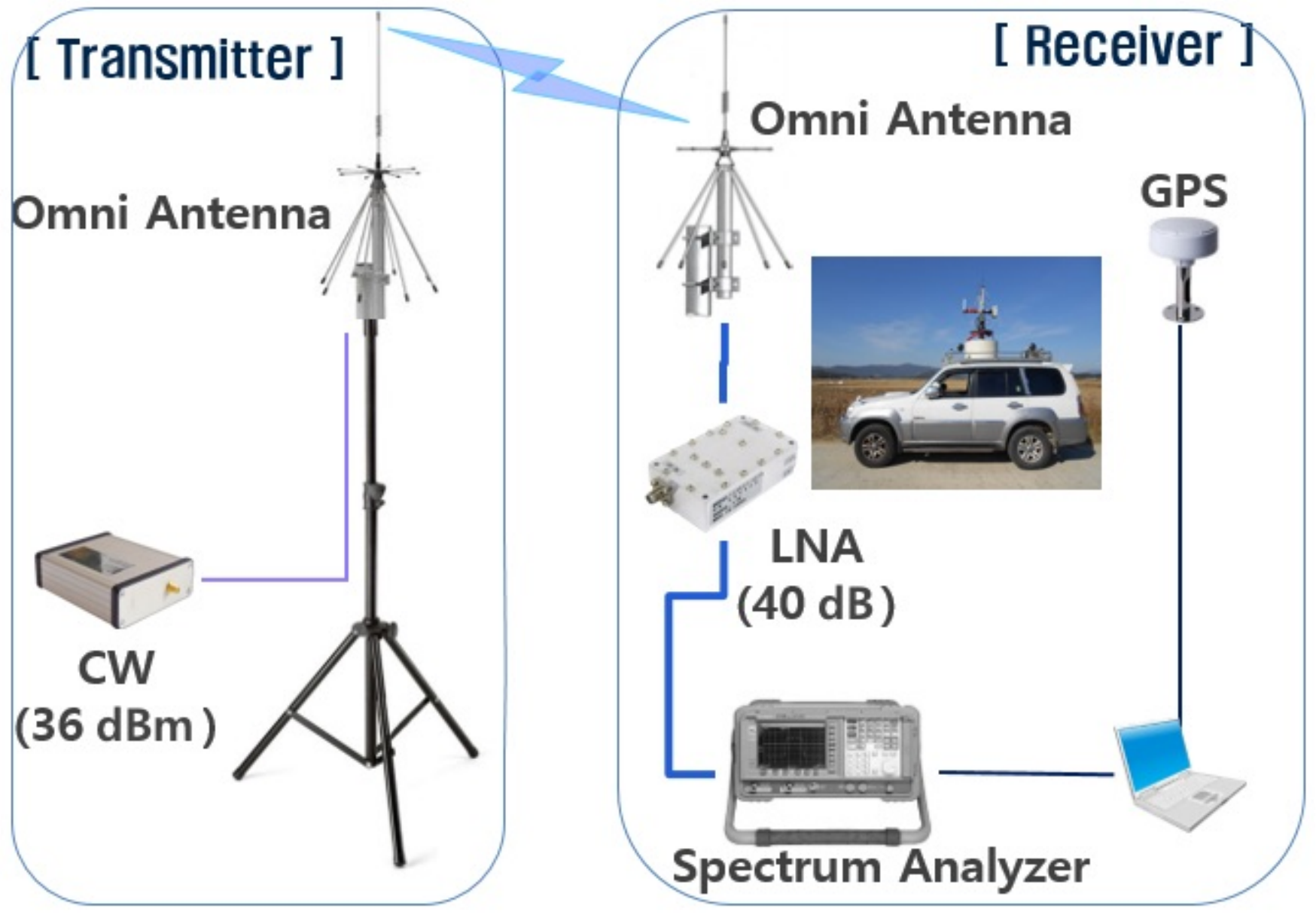}
\caption{Configuration of the measurement system. The transmitter is located on a hill and the receiver is mounted on a vehicle.}
\label{MeasureSys}
\end{figure}
\section{Least-Square Method for Linear Regression}
We use the estimated vector
\begin{equation}
\mathbf{Y}= A + B\mathbf{X_1} + C \mathbf{X_2}
\label{estv}
\end{equation}
for CI path loss application. Here, $A = 32.4 + 20 \log_{10}f_c $, $B = n_{\text{CI}}$, $\mathbf{X_1}= [x_{1,1}, x_{1,2}, \cdots , x_{1,N}]^{\mathrm{T}} =10 \log_{10}\mathbf{d}$ for the distance vector $\mathbf{d} = [d_1, d_2, \cdots , d_N]^{\mathrm{T}}$, $C=\alpha$, and $\mathbf{X_2}= [x_{2,1}, x_{2,2}, \cdots , x_{2,N}]^{\mathrm{T}} = \mathbf{\Phi}$ for the diffraction vector $\mathbf{\Phi} = [\phi_1, \phi_2, \cdots , \phi_N]^{\mathrm{T}}$. The $n$th element $y_n$ of the estimated vector $\mathbf{Y}$ can therefore be written as $y_{n}=A + B x_{1,n} + C x_{2,n}$. 
{Based on the CI model, we set $A$ as a constant and estimate $B$ and $C$ by least-square method, which is given by following theorem:
	\begin{theorem}
	\label{cim}
The solution B, and C is given as follows:
\begin{equation} \label{eq:eq_matrixfinalinvCIdiff_T}
\begin{split}
\begin{pmatrix} B \\ C 
\end{pmatrix}
= &
\begin{pmatrix}
\sum_{n=1}^{N} x_{1,n}  x_{1,n} & \sum_{n=1}^{N} x_{1,n} x_{2,n}  \\
\sum_{n=1}^{N} x_{2,n}  x_{1,n} & \sum_{n=1}^{N} x_{2,n} x_{2,n}  \\
\end{pmatrix}^{-1}
\\
&\begin{pmatrix} 
\sum_{n=1}^{N} x_{1,n}  \widehat{y_{n}} - \sum_{n=1}^{N} A x_{1,n} \\ \sum_{n=1}^{N} x_{2,n}  \widehat{y_{n}} - \sum_{n=1}^{N} A x_{2,} \\
\end{pmatrix},
\end{split}
\end{equation}
for the estimated vector $\mathbf{Y}=[y_1, y_2, \cdots , y_N]^{\mathrm{T}}$ and the measured data $\mathbf{\widehat{Y}} = [\widehat{y_{1}}, \widehat{y_{2}} \cdots, \widehat{y_{N}}]^{\mathrm{T}}$.
		\end{theorem}
		\begin{proof}
See Appendix A.
	\end{proof}}
\begin{figure}[t]
\centering
\includegraphics[width=2.1in]{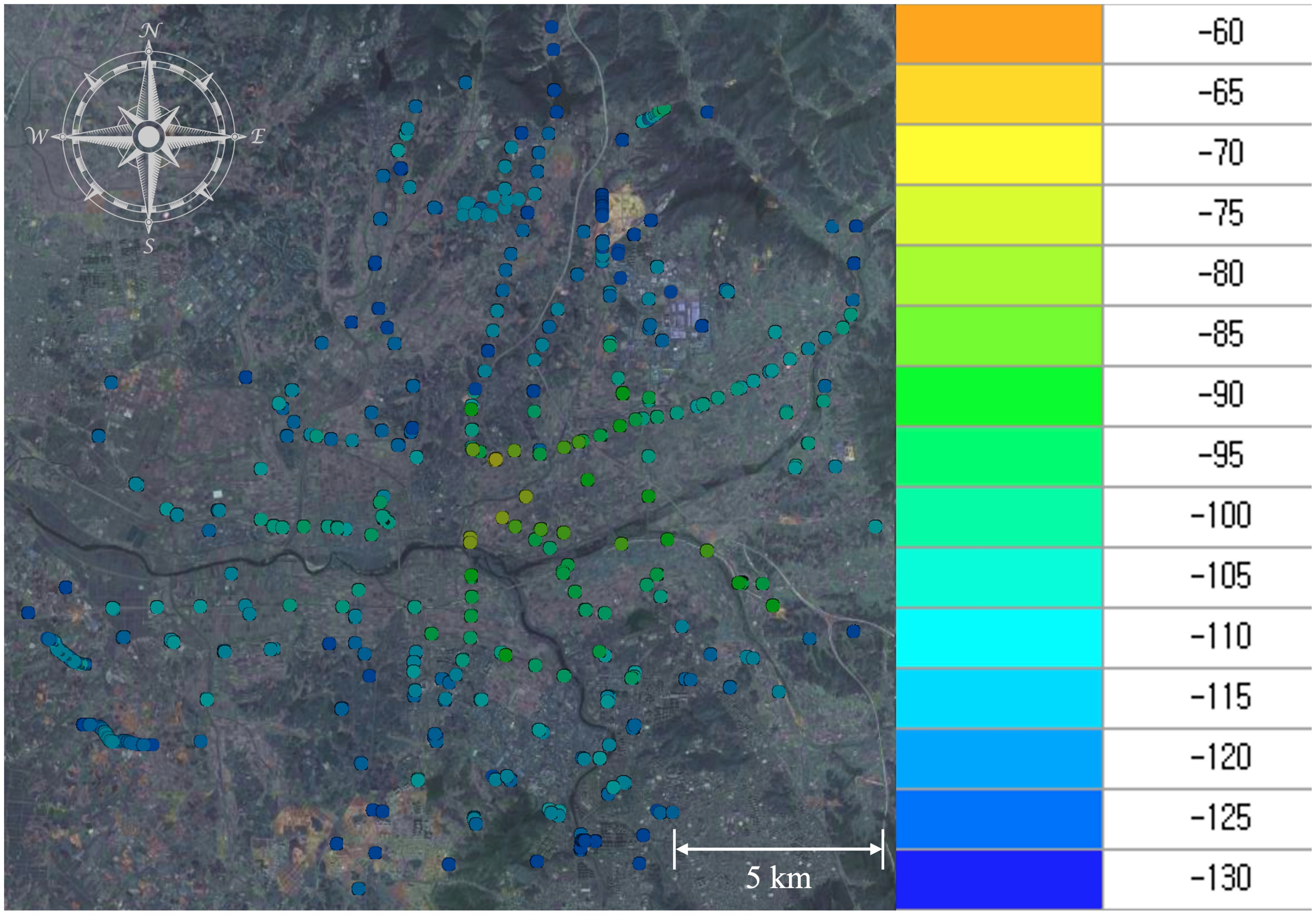}

\caption{Illustration of the measurement positions. The transmitter is located in the center of the map. The colored dots are the measured points of the receiver. The measurement distance ranges up to approximately 12.5 km. Note that most prior work carried out the measurement within a few kms with limited positions. }

\label{IksanMeasureMap}
\end{figure}
\section{Measurement system and analysis}
\subsection{Measurement Scenarios}

The path loss was measured at 1400 and 2250 MHz frequencies, as shown in Fig.~\ref{MeasureSys}. To carry out these measurements, we utilized a system comprising a continuous wave generator, power amplifier, transceiver antenna, low noise amplifier, and spectrum analyzer. For this experiment, we employed omnidirectional transmit and receive antennas and synchronized the global positioning system (GPS) with the measurement system using the RS-232 protocol. The measurements were taken in a rural region of South Korea that includes agricultural land and low hills. The transmitter was placed on a hill, 15 m above ground level, while the receiver was mounted on a vehicle with the antenna at a height of 2 m above ground level. Detailed parameters of the measurement system is given in Table~\ref{tb11}.
\begin{figure}[t]
\centering
\includegraphics[width=2.3in]{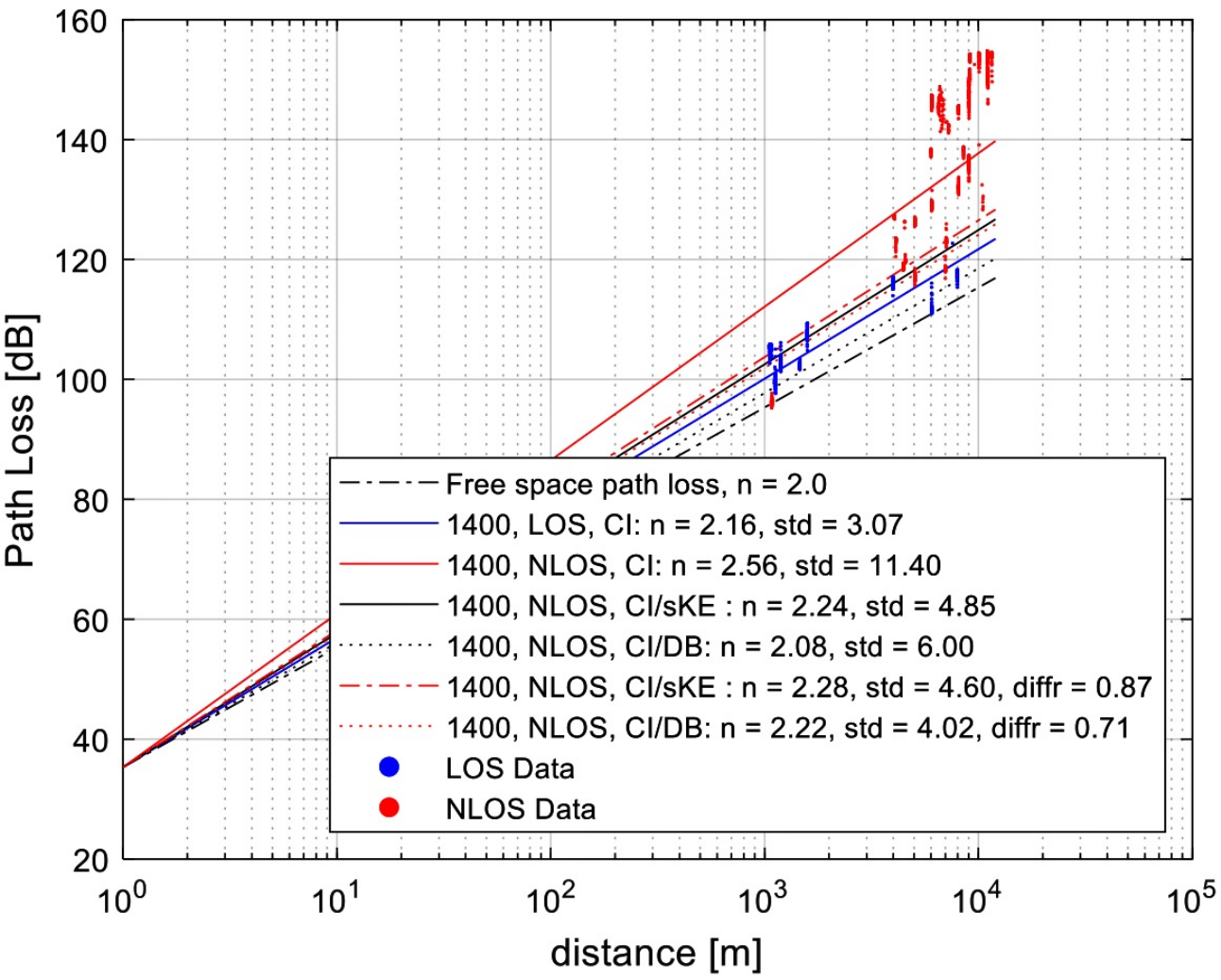}
\caption{Comparison of CI path loss models analyzed for the measurement data in the spring at the frequency of 1400 MHz, where sKE and DB denote single knife edge and delta-Bullington, respectively.}
\label{figpathlossCI1400IksanSpr}
\end{figure}
\begin{table}[t]
\centering
\caption{Parameters of the Measurement System} 	
\label{tb11}
\begin{tabular}{|c|c|}
	\hline
	Parameter&Value \\
	\hline \hline
	Center frequencies & 1400, 2250 (MHz)\\
	\hline
	CW transmitter output  & 36 (dBm) \\
	\hline
	SA bandwidth  & 1 (kHz) \\
	\hline
	LNA  & 40 (dB) \\
	\hline
	Noise level  & -110 (dBm) \\
	\hline
	Effective Reception Level  & -100 (dBm) \\
	\hline
\end{tabular}
\end{table}

The measurement positions and path are illustrated in Fig.~\ref{IksanMeasureMap}, where the dotted circles indicate the measurement locations. The measured power varied between -60 and -130~dBm, and the distance between the transceiver extended up to approximately 12.5~km. It is worth noting that during the power measurements, we assumed that the vehicles were stationary.
\subsection{Analysis Results for CI Path Loss Model}
In this section, we present the analyzed CI path loss results for the measured data based on Sections III and IV.

Fig. \ref{figpathlossCI1400IksanSpr} shows the analyzed path loss models along with the measured path loss data points (scatter plots) at the frequency of 1400 MHz. From the LOS measurement data, the general CI path loss model has a PLE value of 2.16, which is slightly higher than the free-space PLE of 2.00 \cite{MacCartney2017Rural, Sun2016Investi, Rappaport2015Wideband, MacCartney2014Omni, samimi2015Probabilistic, Sulyman2016Directional, rappaport2015millimeter}. With the measured data in NLOS, the model shows a PLE and STD $\chi_\sigma$ of 2.56 and 11.40~dB, figures that are much higher than those of LOS. The original CI path loss models can be summarized as
\begin{multline} \label{eq:eq_1400CIpathlossOutLOS}
\text{PL}_\text{LOS}^{\text{CI}}(f_c,d)[\mathrm{dB}] = 32.4 + 20 \log_{10}f_c + 21.6 \log_{10}d + \chi_{\sigma},\\  \text{LOS measurement},
\end{multline}
\begin{multline} \label{eq:eq_1400CIpathlossOutNLOS}
\text{PL}_\text{NLOS}^{\text{CI}}(f_c,d)[\mathrm{dB}] = 32.4 + 20 \log_{10}f_c + 25.6 \log_{10}d + \chi_{\sigma},\\  \text{NLOS measurement}.
\end{multline}

\begin{table*}[h]
\centering
\caption{{Parameters of PLEs and error statistics of the CI path loss models for the measurement data\\ in a rural area of South Korea in spring at frequency of 1400 and 2250~MHz.}} 	
\label{CIcompareIksan14002250Spring}
\begin{tabular}{|c|c|c|c|c|}
	\hline
	Model&PLE&STD& $\alpha$ & Frequency and note \\
	\hline \hline
	$\text{PL}_{\text{LOS,28}}^{\text{CI-Rural}}$ & 2.1 & 3.6 & - & 28 GHz with Floating-intercept \cite{MacCartney2014Omni}\\
	\hline
	$\text{PL}_{\text{NLOS,28}}^{\text{CI-Rural}}$ & 2.6 & 9.6 &- & 28 GHz with Floating-intercept \cite{MacCartney2014Omni}\\
	\hline
	$\text{PL}_{\text{LOS,73}}^{\text{CI-Rural}}$ & \textbf{2.16} & 1.7 & - & 73 GHz Rural: 14 LOS positions \cite{MacCartney2017Rural}\\
	\hline
	$\text{PL}_{\text{NLOS,73}}^{\text{CI-Rural}}$ & 2.75 & 6.7 &- & 73 GHz Rural: 17 NLOS positions \cite{MacCartney2017Rural}\\
	\hline
	\hline
	$\text{PL}_{\text{LOS}}^{\text{CI-Rural}}$ & \textbf{2.16} & 3.07 &- & 1400 MHz Rural: 1,262 LOS positions \\
	\hline
	$\text{PL}_{\text{NLOS}}^{\text{CI-Rural}}$ & 2.56 & 11.40 &- & 1400 MHz Rural: 2,596 NLOS positions\\
	\hline
	$\text{PL}_{\text{NLOS, sKE}}^{\text{CI-Rural}}$ & 2.24 & 4.85 & 1.00 & 1400 MHz Rural: 2,596 NLOS positions\\
	\hline
	$\text{PL}_{\text{NLOS, DB}}^{\text{CI-Rural}}$ & 2.08 & 6.00 & 1.00 & 1400 MHz Rural: 2,596 NLOS positions\\
	\hline
	$\text{PL}_{\text{NLOS, sKE}}^{\text{CI-Rural}}$ & 2.28 & \textbf{4.60}& \textbf{0.87} & 1400 MHz Rural: 2,596 NLOS positions\\
	\hline
	$\text{PL}_{\text{NLOS, DB}}^{\text{CI-Rural}}$ & 2.22 & \textbf{4.02} & \textbf{0.71} & 1400 MHz Rural: 2,596 NLOS positions\\
	\hline

	$\text{PL}_{\text{LOS}}^{\text{CI-Rural}}$ & \textbf{2.16} & 3.33 &- & 2250 MHz Rural: 1,427 LOS positions\\
	\hline
	$\text{PL}_{\text{NLOS}}^{\text{CI-Rural}}$ & 2.67 & 8.98 &- & 2250 MHz Rural: 3,530 NLOS positions\\
	\hline
	$\text{PL}_{\text{NLOS, sKE}}^{\text{CI-Rural}}$ & 2.27 & 4.92 & 1.00 & 2250 MHz Rural: 3,530 NLOS positions\\
	\hline
	$\text{PL}_{\text{NLOS, DB}}^{\text{CI-Rural}}$ & 2.06 & 5.47 & 1.00 & 2250 MHz Rural: 3,530 NLOS positions\\
	\hline
	$\text{PL}_{\text{NLOS, sKE}}^{\text{CI-Rural}}$ & 2.35 & \textbf{4.42} & \textbf{0.78} & 2250 MHz Rural: 3,530 NLOS positions\\
	\hline
	$\text{PL}_{\text{NLOS, DB}}^{\text{CI-Rural}}$ & 2.26 & \textbf{3.44} & \textbf{0.66} & 2250 MHz Rural: 3,530 NLOS positions\\
	\hline
\end{tabular}
\end{table*}

When we extend the CI model to fully include diffraction component $\phi$ with coefficient $\alpha=1.00$ for the NLOS measurement data, the PLE values reduce to 2.24 and 2.08, while the STDs of the shadow fading dramatically decrease to 4.85 and 6.00, for sKE and DB diffractions, respectively. Here, we just included the diffraction components without regressions. Hence, the CI path loss models with $\alpha=1.00$ can be expressed as
\begin{multline} \label{eq:eq_1400CIpathlossOutsKEDiffull}
\text{PL}_\text{NLOS}^{\text{CI, sKE}}(f_c,d,\phi)[\mathrm{dB}] = 32.4 + 20 \log_{10}f_c  + 22.4 \log_{10}d  \\
+ \textbf{1.00} \cdot \phi + \chi_{\sigma_{\mathrm{sKE}}},  \text{with knife-edge and $\alpha=1.00$},
\end{multline}
\begin{multline} \label{eq:eq_1400CIpathlossOutDBDiffull}
\text{PL}_\text{NLOS}^{\text{CI, DB}}(f_c,d,\phi)[\mathrm{dB}] = 32.4 + 20 \log_{10}f_c + 20.8 \log_{10}d \\
 + \textbf{1.00} \cdot \phi + \chi_{\sigma_{\mathrm{DB}}}, \text{with delta-Bullington and $\alpha=1.00$}.
\end{multline}

After incorporating the regressed coefficient for diffraction $\alpha$ and PLE, we observed that the PLE values changed to 2.28 and 2.22, while $\alpha$ was regressed to 0.87 and 0.71 for sKE and DB, respectively. It is worth noting that the STDs of the shadow fading further decreased to 4.60 and 4.02~dB for sKE and DB, respectively. The analysis results indicate that the PLE values approach 2.00 (the free-space case) when the diffraction component is excluded from the path loss regression. Therefore, it is evident that the modified CI path loss models, which include a diffraction component, are effective in reducing the estimation error. The CI path loss models with the regressed $\alpha$ can be expressed as follows:
\begin{multline} \label{eq:eq_1400CIpathlossOutsKEDiff}
\text{PL}_\text{NLOS}^{\text{CI, sKE}}(f_c,d,\phi)[\mathrm{dB}] = 32.4 + 20 \log_{10}f_c  + 22.8 \log_{10}d  \\
+ \textbf{0.87} \cdot \phi + \chi_{\sigma_{\mathrm{sKE}}},  \text{with knife-edge and regressed $\alpha$},
\end{multline}
\begin{multline} \label{eq:eq_1400CIpathlossOutDBDiff}
\text{PL}_\text{NLOS}^{\text{CI, DB}}(f_c,d,\phi)[\mathrm{dB}] = 32.4 + 20 \log_{10}f_c + 22.2 \log_{10}d \\
 + \textbf{0.71} \cdot \phi + \chi_{\sigma_{\mathrm{DB}}}, \text{with delta-Bullington and regressed $\alpha$}.
\end{multline}

\begin{figure}[t]
\centering
\includegraphics[width=2.3in]{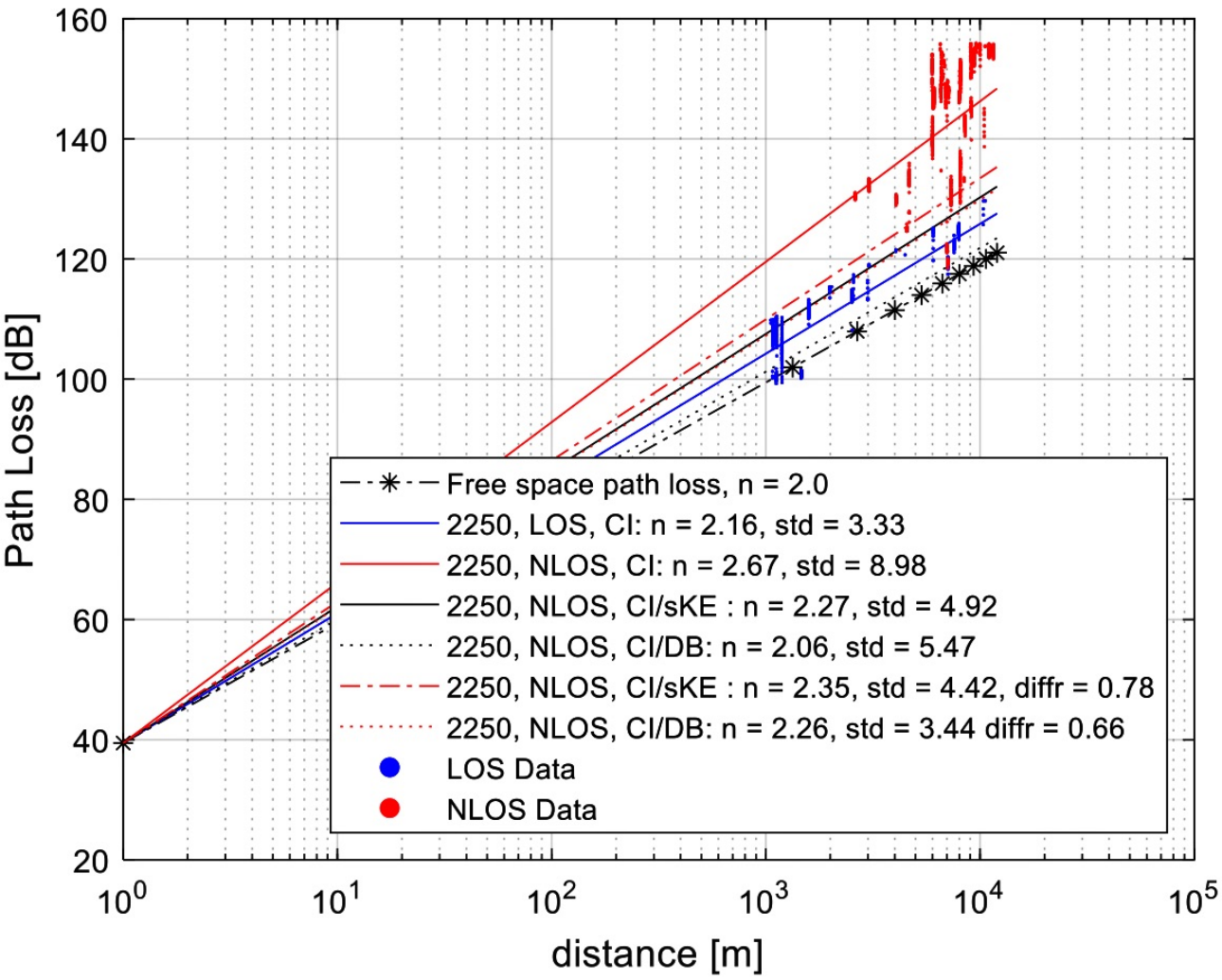}
\caption{Comparison of CI path loss models analyzed for the measurement data in the spring at a frequency of 2250 MHz, where sKE and DB denote single knife edge and delta-Bullington, respectively.}
\label{figpathlossCI2250IksanSpr}
\end{figure}


The path loss models analyzed at a frequency of 2250 MHz, along with the scatter plots of the measured path loss data, are illustrated in Fig.~\ref{figpathlossCI2250IksanSpr}. The CI path loss models with diffraction can be provided from the measured data  at 2250 MHz as
\begin{multline} \label{eq:eq_2250CIpathlossOutsKEDiffull}
\text{PL}_\text{NLOS}^{\text{CI, sKE}}(f_c,d,\phi)[\mathrm{dB}] = 32.4 + 20 \log_{10}f_c  + 22.7 \log_{10}d  \\
+ \textbf{1.00} \cdot \phi + \chi_{\sigma_{\mathrm{sKE}}},  \text{with knife-edge and $\alpha=1.00$},
\end{multline}
\begin{multline} \label{eq:eq_2250CIpathlossOutDBDiffull}
\text{PL}_\text{NLOS}^{\text{CI, DB}}(f_c,d,\phi)[\mathrm{dB}] = 32.4 + 20 \log_{10}f_c + 20.6 \log_{10}d \\
 + \textbf{1.00} \cdot \phi + \chi_{\sigma_{\mathrm{DB}}}, \text{with delta-Bullington and $\alpha=1.00$},
\end{multline}
\begin{multline} \label{eq:eq_2250CIpathlossOutsKEDiff}
\text{PL}_\text{NLOS}^{\text{CI, sKE}}(f_c,d,\phi)[\mathrm{dB}] = 32.4 + 20 \log_{10}f_c  + 23.5 \log_{10}d  \\
+ \textbf{{0.78}} \cdot \phi + \chi_{\sigma_{\mathrm{sKE}}},  \text{with knife-edge and regressed $\alpha$},
\end{multline}
\begin{multline} \label{eq:eq_2250CIpathlossOutDBDiff}
\text{PL}_\text{NLOS}^{\text{CI, DB}}(f_c,d,\phi)[\mathrm{dB}] = 32.4 + 20 \log_{10}f_c + 22.6 \log_{10}d \\
 + \textbf{{0.66}} \cdot \phi + \chi_{\sigma_{\mathrm{DB}}}, \text{with delta-Bullington and regressed $\alpha$}.
\end{multline}

{Furthermore, to enhance the clarity of comparison, we have plotted the path loss data of proposed model and existing data in Fig.~\ref{cidiff}. As shown in Figs.~\ref{figpathlossCI2250IksanSpr} and~\ref{cidiff}, for the original CI path loss model, the LOS PLE value at 2250~MHz is 2.16, which is exactly the same one at 1400~MHz in Fig.~\ref{figpathlossCI1400IksanSpr} and at 73~GHz measurement in~\cite{MacCartney2017Rural}. In case of NLOS path loss models, the PLEs at 2250~MHz with sKE and DB show very small differences as those at 1400~MHz, where the maximum difference is given by 0.07. Additionally, the difference between the measured results in~\cite{MacCartney2017Rural} and~\cite{MacCartney2014Omni} and the results obtained at 1400/2250~MHz is within 0.19, as depicted in Fig~\ref{cidiff}.
In conclusion, the PLE values of CI path loss models at 2250 MHz are nearly the same as those at 1400 MHz, and there is little variation when compared to the measurements taken at 28 and 73~GHz.}
\begin{figure}[t]
\centering
\includegraphics[width=2.2in]{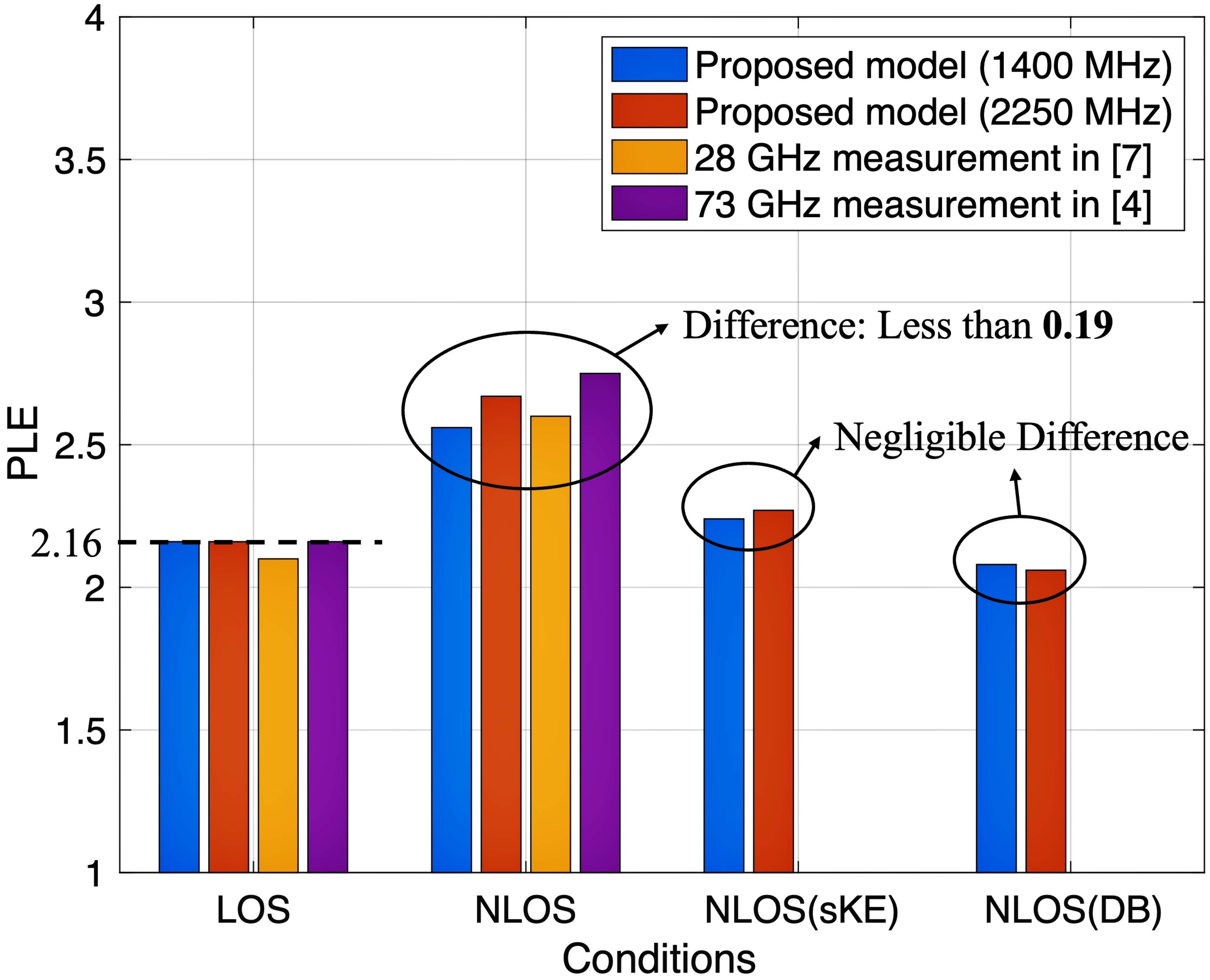}
\caption{Comparison of CI path loss models analyzed for the measurement data at frequency of 1440 and 2250~MHz in proposed model,~\cite{MacCartney2017Rural} and~\cite{MacCartney2014Omni}, where sKE and DB denote single knife edge and delta-Bullington, respectively.}
\label{cidiff}
\end{figure}

{The paper~\cite{MacCartney2017Rural} reported that the 73~GHz measurement yields similar results to the path loss models already specified in 3GPP~\cite{3gpptech}, and can be extrapolated to the frequencies from 0.5 to 100~GHz supported by various measurements along the mmWave bands~\cite{3gppTR38900_2018, 3gppTR38901_2019, 5Gbands, 5Gvtc, Rappaport2015Wideband}. Combining these results, we can draw the conclusion that including a diffraction component in the CI model can considerably enhance the precision of the path loss model and exhibit frequency-independent properties beyond the first meter of propagation, particularly in situations with LOS environments.}

Table~\ref{CIcompareIksan14002250Spring} summarizes the PLE values and STD of the rural CI models at the frequencies of 1400 and 2250~MHz, along with the comparison to the results of reference works \cite{MacCartney2017Rural, MacCartney2014Omni} for millimeter-wave in rural scenarios. It can be observed from the table that the rural LOS PLEs at 1400 and 2250~MHz in South Korea are identical to the rural LOS PLE at 73~GHz in Southwest Virginia, USA. Additionally, our proposed model leads to a significant decrease in STD for NLOS measurements, with values of 4.02 and 3.44~dB for 1400 and 2250~MHz, respectively, compared to 9.6 and 6.7~dB for 28 and 73~GHz scenarios, and 11.40 and 8.98~dB for measurements without considering diffraction.

 
\section{Conclusion}

{This letter presented the large-scale CI path loss models for the RMa scenario in the central rural area of South Korea based on measurements taken during spring.} The measurements were conducted at 1400 and 2250~MHz, and we collected massive data from 3,858 positions (1,262 LOS and 2,596 NLOS) and 4,957 positions (1,427 LOS and 3,530 NLOS) for 1400 and 2250~MHz, respectively. The PLE values of the original CI path loss models at 1400 and 2250~MHz were the same, resulting in 2.16 for LOS and a small difference of 0.11 for NLOS. Compared to the CI path loss models at 73~GHz \cite{MacCartney2017Rural, 3gpptech}, the PLE values were the same as 2.16 for LOS, and the PLE differences were less than 0.19 for NLOS. {We can therefore concluded that, along with the result that the CI model at 73~GHz leads to similar manner with the results in 3GPP~\cite{MacCartney2017Rural}, the CI path loss model has frequency-independent characteristics beyond the first meter of propagation, and hence we can deduce that our proposed CI models may be used for frequencies up to millimeter-wave~\cite{map, IAB, HBFD, UAVRIS, yhtvt}.}


When we extended the CI models to include diffraction loss component without regression, in case of 1400~MHz, the PLE values reduced to 2.24 and 2.08,  and furthermore, the STD dramatically decreased to 4.85 and 6.00 for the sKE and DB, respectively. When we included the diffraction loss component in regression, the PLE values changed to 2.28 and 2.22, and the STD further decreased to 4.60 dB and 4.02 dB for sKE and DB, respectively. In case of 2250~MHz, the PLE and STD values show a very similar manner to those of 1400~MHz, which is also shown in Fig.~\ref{figpathlossCI1400IksanSpr} and \ref{figpathlossCI2250IksanSpr}. According to the derived results, it is noted that the PLE values approached 2 when the diffraction component is isolated from the regression of overall path loss. {Therefore, we concluded that introducing the diffraction component CI model can dramatically contribute to modeling the path loss more accurately.}
		\section*{Appendix A}
\section*{Proof of Theorem~\ref{cim}}
{The cost function $\mathcal{J}$ is given by the mean-squared error (MSE) between the estimated vector $\mathbf{Y}=[y_1, y_2, \cdots , y_N]^{\mathrm{T}}$ and the measured data  $\mathbf{\widehat{Y}} = [\widehat{y_{1}}, \widehat{y_{2}}, \cdots, \widehat{y_{N}}]^{\mathrm{T}}$: 
\begin{equation} \label{eq:eq_cost}
\mathcal{J} = \frac{1}{N} \lVert  \mathbf{\widehat{Y}} - \mathbf{Y}  \rVert^2=\frac{1}{N} \sum_{n=1}^{N} \left( \widehat{y_{n}} - y_{n}  \right)^2.
\end{equation}
Therefore, minimizing $\mathcal{J}$ is equivalent to the least-square problem as
\begin{equation} \label{eq:eq_mincost}
\min_{A,B,C,D,...} \lVert  \mathbf{\widehat{Y}} - \mathbf{Y}  \rVert^2 = \min_{A,B,C,D,...} \sum_{n=1}^{N} \left( \widehat{y_{n}} - y_{n}  \right)^2.
\end{equation}
Differentiating the objective function of~(\ref{eq:eq_mincost}) with respect to $B$ and $C$ and giving the result to zero, we can get
\begin{equation} \label{eq:eq_partialBcostResultforCI}
\begin{split}
\sum_{n=1}^{N} x_{1,n}  \widehat{y_{n}} = \sum_{n=1}^{N} A x_{1,n} +  \sum_{n=1}^{N}B x_{1,n} x_{1,n} + \sum_{n=1}^{N} C x_{1,n} x_{2,n},\\
\sum_{n=1}^{N} x_{2,n}  \widehat{y_{n}} = \sum_{n=1}^{N} A x_{2,n} +  \sum_{n=1}^{N}B x_{2,n} x_{1,n} + \sum_{n=1}^{N} C x_{2,n} x_{2,n},
\end{split}
\end{equation}
respectively. Since $A$ is a constant and not to be regressed, we can modify ($\ref{eq:eq_partialBcostResultforCI}$) into a linear system with respect to $B$ and $C$ by
\begin{equation} \label{eq:eq_matrixfinalCIfiff}
\begin{split}
&\begin{pmatrix}
\sum_{n=1}^{N} x_{1,n}  x_{1,n} & \sum_{n=1}^{N} x_{1,n} x_{2,n}  \\
\sum_{n=1}^{N} x_{2,n}  x_{1,n} & \sum_{n=1}^{N} x_{2,n} x_{2,n}  \\
\end{pmatrix}
\begin{pmatrix}  B \\ C 
\end{pmatrix}
\\
&= \begin{pmatrix} \sum_{n=1}^{N} x_{1,n}  \widehat{y_{n}} - \sum_{n=1}^{N} A x_{1,n} \\ \sum_{n=1}^{N} x_{2,n}  \widehat{y_{n}} - \sum_{n=1}^{N} A x_{2,n} \\ 
\end{pmatrix}.
\end{split}
\end{equation}
Hence, by applying matrix inverse to the left-hand side of~(\ref{eq:eq_matrixfinalCIfiff}), the least-square solution of~(\ref{eq:eq_mincost}) is given by
\begin{equation} \label{eq:eq_matrixfinalinvCIdiff}
\begin{split}
\begin{pmatrix} B \\ C 
\end{pmatrix}
= &
\begin{pmatrix}
\sum_{n=1}^{N} x_{1,n}  x_{1,n} & \sum_{n=1}^{N} x_{1,n} x_{2,n}  \\
\sum_{n=1}^{N} x_{2,n}  x_{1,n} & \sum_{n=1}^{N} x_{2,n} x_{2,n}  \\
\end{pmatrix}^{-1}
\\
&\begin{pmatrix} 
\sum_{n=1}^{N} x_{1,n}  \widehat{y_{n}} - \sum_{n=1}^{N} A x_{1,n} \\ \sum_{n=1}^{N} x_{2,n}  \widehat{y_{n}} - \sum_{n=1}^{N} A x_{2,} \\
\end{pmatrix},
\end{split}
\end{equation}
and the theorem follows.~ $\blacksquare$}
\ifCLASSOPTIONcaptionsoff
  \newpage
\fi



\bibliographystyle{IEEEtran}
\bibliography{IEEEabrv,JDREFm}

\begin{thebibliography}{10}
\providecommand{\url}[1]{#1}
\csname url@samestyle\endcsname
\providecommand{\newblock}{\relax}
\providecommand{\bibinfo}[2]{#2}
\providecommand{\BIBentrySTDinterwordspacing}{\spaceskip=0pt\relax}
\providecommand{\BIBentryALTinterwordstretchfactor}{4}
\providecommand{\BIBentryALTinterwordspacing}{\spaceskip=\fontdimen2\font plus
\BIBentryALTinterwordstretchfactor\fontdimen3\font minus
  \fontdimen4\font\relax}
\providecommand{\BIBforeignlanguage}[2]{{%
\expandafter\ifx\csname l@#1\endcsname\relax
\typeout{** WARNING: IEEEtran.bst: No hyphenation pattern has been}%
\typeout{** loaded for the language `#1'. Using the pattern for}%
\typeout{** the default language instead.}%
\else
\language=\csname l@#1\endcsname
\fi
#2}}
\providecommand{\BIBdecl}{\relax}
\BIBdecl

\bibitem{3gppTR36873_2017}
3GPP, ``Study on {3D} channel model for {LTE},'' \emph{3rd Generation
  Partnership Project (3GPP) Technical Report}, vol. TR 36.873 Rel. 12, 2017.

\bibitem{ITURM2135_2009guidelines}
ITU, ``Guidelines for evaluation of radio interface technologies for
  {IMT}-advanced,'' \emph{ITU-R M Series}, vol. ITU-R M.2135-1, 2009.

\bibitem{3gppTR38900_2018}
3GPP, ``Study on channel model for frequency spectrum above 6 {GHz},''
  \emph{3rd Generation Partnership Project (3GPP) Technical Report}, vol. TR
  38.900 Rel. 15, 2018.

\bibitem{MacCartney2017Rural}
G.~R. {MacCartney} and T.~S. {Rappaport}, ``Rural macrocell path loss models
  for millimeter wave wireless communications,'' \emph{IEEE J. Sel. Areas
  Commun.}, vol.~35, no.~7, pp. 1663--1677, Jul. 2017.

\bibitem{3gppTR38901_2019}
3GPP, ``Study on channel model for frequencies from 0.5 to 100 {GHz},''
  \emph{3rd Generation Partnership Project (3GPP) Technical Report}, vol. TR
  38.901 Rel. 16, 2019.

\bibitem{rappaport2002wireless}
T.~S. Rappaport \emph{et~al.}, \emph{Wireless communications: principles and
  practice}.\hskip 1em plus 0.5em minus 0.4em\relax Prentice Hall PTR New
  Jersey, 2002, vol.~2.

\bibitem{MacCartney2014Omni}
G.~R. {MacCartney} \emph{et~al.}, ``Omnidirectional path loss models in {New
  York City} at 28 {GHz} and 73 {GHz},'' in \emph{Proc. IEEE Int. Symp. Pers.,
  Ind., and Mob. Radio Commun. (PIMRC)}, 2014, pp. 227--231.

\bibitem{samimi2015Probabilistic}
M.~K. Samimi \emph{et~al.}, ``Probabilistic omnidirectional path loss models
  for millimeter-wave outdoor communications,'' \emph{IEEE Wireless Commun.
  Lett.}, vol.~4, no.~4, pp. 357--360, Aug. 2015.

\bibitem{Sun2016Investi}
S.~{Sun} \emph{et~al.}, ``Investigation of prediction accuracy, sensitivity,
  and parameter stability of large-scale propagation path loss models for {5G}
  wireless communications,'' \emph{IEEE Trans. Veh. Technol.}, vol.~65, no.~5,
  pp. 2843--2860, May 2016.

\bibitem{ITURP526_2013}
ITU, ``Propagation by diffraction,'' \emph{ITU-R P Series}, vol. ITU-R
  P.526-13, 2013.

\bibitem{ITURP1812_2019}
------, ``A path-specific propagation prediction method for point-to-area
  terrestrial services in the {VHF} and {UHF} bands,'' \emph{ITU-R P Series},
  vol. ITU-R P.1812-5, 2019.

\bibitem{Rappaport2015Wideband}
T.~S. {Rappaport} \emph{et~al.}, ``Wideband millimeter-wave propagation
  measurements and channel models for future wireless communication system
  design,'' \emph{IEEE Trans. Commun.}, vol.~63, no.~9, pp. 3029--3056, Sep.
  2015.

\bibitem{3gpptech}
3GPP, ``Technical specification group radio access network; channel model for
  frequency spectrum above 6 {GHz} (release 14),'' \emph{3rd Generation
  Partnership Project (3GPP) Technical Report}, vol. TR 38.900 V14.2.0, 2019.

\bibitem{3gpp24}
------, ``New measurements at 24 {GHz} in a rural macro environment,''
  \emph{Telstra, Ericsson}, vol. DOC R1-164975,, 2016.

\bibitem{metis}
METIS, ``Metis channel model metis2020,'' \emph{Deliverable D1.4 v3}, 2015.

\bibitem{ITU23}
\BIBentryALTinterwordspacing
ITU, ``{ITU-R} preparatory studies for {WRC}-23.'' [Online]. Available:
  \url{https://www.itu.int/en/ITU-R/study-groups/rcpm/Pages/wrc-23-studies.aspx}
\BIBentrySTDinterwordspacing

\bibitem{FAITU}
\BIBentryALTinterwordspacing
------, ``{WRC}-19, provisional final acts.'' [Online]. Available:
  \url{https://www.itu.int/dms\_pub/itu-r/opb/act/R-ACT-WRC.14-2019-PDF-E.pdf}
\BIBentrySTDinterwordspacing

\bibitem{HBFSO}
H.-B. Jeon \emph{et~al.}, ``Free-space optical communications for {6G} wireless
  networks: Challenges, opportunities, and prototype validation,'' \emph{IEEE
  Commun. Mag.}, vol.~61, no.~4, pp. 116--121, Apr. 2023.

\bibitem{hjtvt}
H.-J. Moon \emph{et~al.}, ``Performance analysis of passive retro-reflector
  based tracking in free-space optical communications with pointing errors,''
  \emph{IEEE Trans. Veh. Technol.}, pp. 1--6, Mar. 2023.

\bibitem{map}
Y.-G. {Lim} \emph{et~al.}, ``Map-based millimeter-wave channel models: An
  overview, data for {B5G} evaluation and machine learning,'' \emph{IEEE
  Wireless Commun.}, vol.~27, no.~4, pp. 54--62, Aug. 2020.

\bibitem{IAB}
G.~Y. Suk \emph{et~al.}, ``Full duplex integrated access and backhaul for {5G
  NR}: Analyses and prototype measurements,'' \emph{IEEE Wireless Commun.},
  vol.~29, no.~4, pp. 40--46, Aug. 2022.

\bibitem{SA}
Z.~A. Shamsan, ``Rainfall and diffraction modeling for millimeter-wave wireless
  fixed systems,'' \emph{IEEE Access}, vol.~8, pp. 212\,961--212\,978, 2020.

\bibitem{CIPL}
S.~Sun \emph{et~al.}, ``Path loss, shadow fading, and line-of-sight probability
  models for {5G} urban macro-cellular scenarios,'' in \emph{Proc. IEEE IEEE
  Glob. Commun. Conf. Workshop (GCW)}, 2015, pp. 1--7.

\bibitem{Sulyman2016Directional}
A.~I. {Sulyman} \emph{et~al.}, ``Directional radio propagation path loss models
  for millimeter-wave wireless networks in the 28-, 60-, and 73-{GHz} bands,''
  \emph{IEEE Trans. Wireless Commun.}, vol.~15, no.~10, pp. 6939--6947, Oct.
  2016.

\bibitem{rappaport2015millimeter}
T.~S. Rappaport \emph{et~al.}, \emph{Millimeter wave wireless
  communications}.\hskip 1em plus 0.5em minus 0.4em\relax Pearson Education,
  2015.

\bibitem{5Gbands}
\BIBentryALTinterwordspacing
A.~Univ. \emph{et~al.}, ``{5G} channel model for bands up to 100 {GHz},'' Oct.
  2016. [Online]. Available: \url{http://www.5gworkshops.com/5GCM.html}
\BIBentrySTDinterwordspacing

\bibitem{5Gvtc}
K.~Haneda \emph{et~al.}, ``Indoor {5G 3GPP-like} channel models for office and
  shopping mall environments,'' in \emph{Proc. IEEE Int. Conf. Commun.
  Workshops (ICCW)}, 2016, pp. 694--699.

\bibitem{HBFD}
H.-B. Jeon \emph{et~al.}, ``Graph-theory-based resource allocation and mode
  selection in {D2D} communication systems: The role of full-duplex,''
  \emph{IEEE Wireless Commun. Lett.}, vol.~10, no.~2, pp. 236--240, Feb. 2021.

\bibitem{UAVRIS}
------, ``An energy-efficient aerial backhaul system with reconfigurable
  intelligent surface,'' \emph{IEEE Trans. Wireless Commun.}, vol.~21, no.~8,
  pp. 6478--6494, Aug. 2022.

\bibitem{yhtvt}
Y.~Kim \emph{et~al.}, ``Partition-based {RIS}-assisted multiple access: {NOMA}
  decoding order perspective,'' \emph{IEEE Trans. Veh. Technol.}, vol.~71,
  no.~8, pp. 9083--9088, Aug. 2022.

\end{thebibliography}
\end{document}